\documentclass[12pt,a4paper]{article}

\usepackage{amsmath,amssymb,amsfonts,amsthm,bbm}
\usepackage{lmodern}
\usepackage[T1]{fontenc}
\usepackage{fullpage}

\usepackage[pdftex]{graphicx}
\usepackage{caption}
\usepackage{subfig}
\usepackage{color}
\usepackage{float}
\restylefloat{figure}

\usepackage{tikz-cd}
\usetikzlibrary{matrix,arrows,decorations.pathmorphing}

\theoremstyle{plain}
\newtheorem{theorem}{Theorem}[section]
\newtheorem{proposition}[theorem]{Proposition}

\theoremstyle{definition}

\theoremstyle{remark}

\newcommand{\Her}{\operatorname{Her}}

\newcommand{\HH}{\mathcal{H}}

\newcommand{\D}{\mathcal{D}}

\newcommand{\U}{\mathcal{U}}

\renewcommand{\u}{\mathfrak{u}}

\newcommand{\T}{\operatorname{T}\!}

\newcommand{\Ker}{\operatorname{Ker}}
\newcommand{\Ad}{\operatorname{Ad}}

\newcommand{\Tr}{\operatorname{Tr}}

\newcommand{\diag}{\operatorname{diag}}

\renewcommand{\phi}{\varphi}

\newcommand{\1}{{\mathbf 1}}
\newcommand{\0}{{\mathbf 0}}

\newcommand{\obs}[1]{\hat{#1}}

\begin{document}
\title{A geometric framework   for mixed quantum states  based on a  K\"{a}hler structure }

\author{Hoshang Heydari\\\emph{ Department of Physics, Stockholm  University, SE-106 91 Stockholm, Sweden}}

\maketitle

\begin{abstract}
In this paper we introduce a   geometric framework for mixed quantum states based on a K\"ahler structure. The geometric framework includes a symplectic form, an almost complex structure, and a Riemannian metric that characterize the space of mixed quantum states. We  argue that the almost complex structure is integrable.  We also in detail  discuss a visualizing application of this geometric framework by deriving a geometric uncertainty relation for mixed quantum states. The framework is computationally effective  and it provides us with a better understanding of general quantum mechanical systems.
\end{abstract}

\section{Introduction}
Geometric quantum mechanics describes quantum mechanical systems based on their underlying geometrical structures \cite{Gunter_1977,Kibble_1979,Ashtekar_etal1998,Brody_etal1999}. Recently, it has been shown that such geometrical structures of quantum theory have profound information about foundations and the nature of the theory with many applications in quantum science and technology \cite{Zanardi_etal1999,Ekert_etal2000,Solinas_etal2003,Uhlmann1989,Uhlmann1991}.

In geometric quantum mechanics the projective Hilbert space is constructed by general Hopf fibration of hypersphere and usually is called the quantum phase space of a pure quantum state. However, a pure state is a very limited class of quantum states, namely mixed quantum states. We know a lot about the geometry of pure quantum states but our knowledge are  very limited when we consider mixed quantum states.

Recently, we have introduced a geometric framework for density operators based on fiber bundles which has lead to many interesting results such as a geometric phases, an uncertainty relations, quantum speed limits, a distance measure, and an optimal Hamiltonian \cite{GP,MB,DD,GQE,GUR,QSL}. Note the geometric framework that we introduce in this paper is different from the fiber bundles one.

In this paper we introduce a  geometric framework for mixed quantum states based on a specific K\"ahler structure. The mathematical structure are well-known in the mathematical literature,  but is almost unknown to physicists. In  section \ref{s1} we introduce the geometric framework for mixed quantum states. We will in detail discuss  Kirillov-Kostant-Souriau K\"ahler structure and  the existence of an almost complex structure on quantum phase space of mixed states. We will also briefly discuss integrability of the almost complex structure. In section \ref{s2} we will apply our geometric framework to quantum systems by deriving a geometric uncertainty relation for mixed quantum states which is one of the most important topics that distinguish quantum physics from classical physics \cite{GUR}.

\section{Geometric framework}\label{s1}
There are three  important geometries. The most well-known one is called Riemannian geometry which is defined to be the geometry of a positive-definite symmetric bilinear form. The Riemannian geometry is a well-developed subject and we will not further discuss it here in this text. Moreover, the geometry of a closed non-degenerate skew-symmetric bilinear form  is called symplectic geometry. Finally, the geometry  of a linear bundle map with square -1 is called almost complex geometry.
A K\"{a}hler manifold is symplectic manifold which is equipped  with an integrable almost complex structure.

In this section we introduce a new geometric framework for general finite dimensional quantum systems based on  a specific K\"{a}hler structure which is called Kirillov-Kostant-Souriau K\"ahler structure
. In the following text we will denote the identity map by
$\1$ and we let $\1_{n}$ be  the $n\times n$ identity matrix, and $\0_{n}$ is the $n\times n$ zero matrix.
\subsection{The Kirillov-Kostant-Souriau K\"ahler structure}
In the first step we will define Kirillov-Kostant-Souriau K\"ahler structure for the space of density operators.
To do so we let  $\HH$ be an $n$-dimensional Hilbert space, $\U(\HH)$ be the group of unitary operators on $\HH$,
$\Her(\HH)$ be the space of Hermitian operators on $\HH$, and  the adjoint action of $\U(\HH)$ on $\Her(\HH)$
\begin{equation}
\U(\HH)\times \Her(\HH)\longrightarrow \Her(\HH),
\end{equation}
defined by
\begin{equation}\label{conj}
(U,\obs{A})\mapsto \Ad_U(\obs{A})=U\obs{A}U^\dagger.
\end{equation}
The manifold $\Her(\HH)$ is diffeomorphic to the homogeneous space $U(n)/U(n_{1})\times U(n_{2})\times\cdots \times U(n_{k})$. It is easy to show that  $\Her(\HH)$ is a flag manifold.
A density operator on $\HH$ is a member of $\Her(\HH)$ whose eigenvalues are non-negative and sum up to $1$.
We write $\D(\HH)$ for the space of density operators on $\HH$.
Note that, the adjoint action preserves $\D(\HH)$, and the orbits of the action in $\D(\HH)$ are in one-to-one correspondence with the possible spectra for density operators on $\HH$.
To be precise, two density operators belong to the same orbit if and only if they have the same spectrum.
Given such a spectrum $\sigma$, we write $\D(\sigma)$ for the corresponding orbit.
In this section we introduce an $\Ad$-equivariant K\"ahler structure on $\D(\sigma)$ called the Kirillov-Kostant-Souriau K\"ahler structure \cite{KKS}.
We remind the reader that a K\"ahler structure is pair $(\omega,J)$ consisting of a symplectic structure $\omega$ and a complex structure $J$, and that associated to such a structure is a Hermitian inner product,
\begin{equation}
h(X,Y)=\omega(X,JY)+i\omega(X,Y).
\end{equation}
Note also that $(X,Y)\longmapsto\omega(X,JY)$ is a Riemannain metric on $\D(\HH)$.

\subsection{Representation of tangent vectors}

Next we want to define  representations of  tangent vectors on the orbit of the adjoint action.
Note  that the adjoint action \eqref{conj} is transitive, that is, for each density operator $\rho$ we have a surjective linear map $\Lambda_\rho:\Her(\HH)\to\T_\rho\D(\sigma)$ defined by
\begin{equation}
\Lambda_\rho(\obs{H})=\frac{1}{i\hbar}[\hat{H},\rho]
\end{equation}
since any elements in $\T_\rho\D(\sigma)$ can be written as $[\hat{H},\rho]$.
Note that, since the map $(X,Y)\longmapsto \mathrm{Tr}(XY)$ define a bilinear form on $\mathfrak{u}(n)$ which is non-degenerated and invariant under conjugation, the kernel of $\Lambda_\rho(\obs{H})$ is a subspace of $\mathfrak{u}(n)$ which is the Lie algebra of the stabilizer  of $\rho$ for the group action $U(n)$. We can also identify the Lie algebra $\mathfrak{u}(n)$ with its dual $\mathfrak{u}^{*}(n)$ which implies that the $U(n)$ action on $\mathfrak{u}(n)$ or $H$ is adjoint or co-adjoint action. Thus $\Her(\HH)$ cán be described by co-adjoint of  $\mathfrak{u}(n)$.
The kernel of $\Lambda_\rho$ consists of all Hermitian operators on $\HH$ that commutes with $\rho$,
and we define a complementary space to $\Ker\Lambda_\rho$ as follows.

Let $p_1>p_2>\dots>p_k$ be the different eigenvalues in the spectrum of the density operator, $\sigma$, and $n_j$ be the multiplicity of $p_j$.
We can always find a basis in $\HH$ relative which
\begin{equation}
\rho=\diag(p_1\1_{n_1},p_2\1_{n_2},\dots,p_k\1_{n_k}).
\end{equation}
Moreover, the kernel of $\Lambda_{\rho}$ consists of all those Hermitian operators $\obs{A}$ which are represented by block diagonal matrices
\begin{equation}
\obs{A}=\diag(A_{11},A_{22},\dots, A_{kk}),
\end{equation}
relative to this basis where each $A_{jj}$ is an $n_j\times n_j$ Hermitian matrix. We define the complementary space $\Ker\Lambda_{\rho}^\bot$ to consist of all the Hermitian operators that are represented by off-diagonal matrices
\begin{equation}\label{off}
\obs{B}=\left[
  \begin{array}{ccccc}
  \0_{n_1}        & B_{12}         & B_{13}  & \ldots & B_{1k}  \\
  B_{12}^\dagger & \0_{n_2}        & B_{23}  & \ldots & B_{2k}  \\
  B_{13}^\dagger & B_{23}^\dagger & \0_{n_3} & \ldots & B_{3k}  \\
  \vdots         & \vdots         & \vdots  & \ddots & \vdots  \\
  B_{1k}^\dagger & B_{2k}^\dagger & B_{3k}^\dagger   & \ldots  & \0_{n_k}
  \end{array}
\right].
\end{equation}
Obviously, $\Her(\HH)=\Ker\Lambda_{\rho}\oplus \Ker\Lambda_{\rho}^\bot$, and $\Lambda_\rho$ maps $\Ker\Lambda_{\rho}^\bot$ isomorphically onto $\T_\rho\D(\sigma)$. Now we are in right position to define an almost complex structure on quantum phase space.

\subsection{Almost complex structure}
An almost complex structure on a manifold is an automorphism of its tangent bundle whose square equals $-\1$.
Moreover, the almost complex structure is  a complex structure if it is integrable, meaning that a rank two tensor, usually called the Nijenhuis tensor vanishes. We will discuss integrability of almost complex structure in the following text.
Note also that manifolds that admit complex structures can be equipped with holomorphic atlases.
That is, they are complex manifolds.

The orbit $\D(\sigma)$ does admit an Ad-invariant complex structure $J$;
we define an operator $\obs{B}\mapsto\check{B}$ on $\ker\Lambda_\rho$, where, if $\obs{B}$ is given by \eqref{off}, the operator $\check{B}$
is given by
\begin{equation}
\check{B}=\left[
  \begin{array}{ccccc}
  \0_{n_1}        & iB_{12}         & iB_{13}  & \ldots & iB_{1k}  \\
  -iB_{12}^\dagger & \0_{n_2}        & iB_{23}  & \ldots & iB_{2k}  \\
  -iB_{13}^\dagger & -iB_{23}^\dagger & \0_{n_3} & \ldots & iB_{3k}  \\
  \vdots         & \vdots         & \vdots  & \ddots & \vdots  \\
  -iB_{1k}^\dagger & -iB_{2k}^\dagger & -iB_{3k}^\dagger   & \ldots  & \0_{n_k}
  \end{array}
\right].
\end{equation}
Now, the bundle map $J:\T\D(\sigma)\to\T\D(\sigma)$, defined by
\begin{equation}
J\left(\frac{1}{i\hbar}[\obs{B},\rho]\right)=\frac{1}{i\hbar}[\check{B},\rho],
\end{equation}
where  $J\left(\frac{1}{i\hbar}[\obs{B},\rho]\right)= \frac{1}{i\hbar}[j(\obs{B}),\rho]$, e.g.,  for the matrix $\hat{B}=(B_{kl})$ we have $j(\hat{B})= (i B_{kl})$. Note that $J$ satisfies $J^2=-\1$, as follows
\begin{equation}
J\left(J\left(\frac{1}{i\hbar}[\obs{B},\rho]\right)\right)=J\left(\frac{1}{i\hbar}
[j(\hat{B}),\rho]\right)=\frac{1}{i\hbar}
[j(j(\hat{B})),\rho]=
\frac{1}{i\hbar}[-\obs{B},\rho]\Longrightarrow J^{2}=-1,
\end{equation}
and thus is an almost complex structure.
Next we show that $J$ is integrable, and hence is a complex structure.

\subsection{Integrability of $J$ on quantum phase space}
We have derived an almost complex structure  for the quantum phase space of mixed states. One important  question concerning this almost complex structure is the integrability of $J$ which we will investigate in this section. An intergrade almost structure has the structure
of a complex analytic manifold.  Let $J$ be an almost complex structure on our quantum phase space. A condition for integrability of $J$ is the following. We can associate a $(2,1)$-tensor $N^{J}$ defined by
\begin{equation}
N^{J}(X,Y)=[X,Y]+J[JX,Y]+J[X,JY]-[JX,JY],
\end{equation}
for all $X,Y\in T\mathcal{D}(\sigma)$ is a local vector fields, to every almost complex structure $J$.  $N^{J}(X,Y)$ is  called Nijenhuis tensor. Then we have the following proposition:
\begin{proposition}
Let $J$ be an almost complex structure on our quantum phase space $\mathcal{D}(\sigma)$. Then these two statements are equivalent
\begin{enumerate}
  \item $J$ be an almost complex structure
  \item $N^{J}=0$.
\end{enumerate}
\end{proposition}
For the  proof and more information see \cite{Dasilva,Mein,Gompf}. In the next section we define the most important structure of the geometric framework.

\subsection{K\"ahler structure}
In this section we define Kirillov-Kostant-Souriau symplectic form and derive an explicit expression for Hermitian inner product on the quantum phase space $\D(\sigma)$.
The Kirillov-Kostant-Souriau symplectic form on $\D(\sigma)$ is defined by
\begin{equation}\label{KKSsymp}
\omega\left(\frac{1}{i\hbar}[\obs{A},\rho],\frac{1}{i\hbar}[\obs{B},\rho]\right)
=\frac{1}{i\hbar}\Tr\left([\hat{A},\hat{B}]\rho\right)
=\frac{1}{i\hbar}\Tr\left(\hat{A}[\hat{B},\rho]\right).
\end{equation}
\begin{theorem}
The symplectic form $\omega$ (\ref{KKSsymp}) is non-degenerated  and closed.
\end{theorem}
\begin{proof}
The symplectic form $\omega$ is non-degenerated since if we chose $\hat{A}=[\hat{B},\rho]$ in equation (\ref{KKSsymp}) then $\Tr\left(\hat{A}[\hat{B},\rho]\right)\neq0$ which implies that $\omega\left(\frac{1}{i\hbar}[\obs{A},\rho],\frac{1}{i\hbar}[\obs{B},\rho]\right)\neq0$.
Next we will prove that the symplectic form $\omega\left(\frac{1}{i\hbar}[\obs{A},\rho],\frac{1}{i\hbar}[\obs{B},\rho]\right)$ is closed, that is
\begin{equation}
d\omega\left(\frac{1}{i\hbar}[\obs{A},\rho],\frac{1}{i\hbar}[\obs{B},\rho],
\frac{1}{i\hbar}[\obs{C},\rho]\right)=0,
\end{equation}
for all $\hat{A},\hat{B},\hat{C}\in \mathfrak{u}(n)$  as follows. Let $\overline{\hat{A}}$, $\overline{\hat{B}}$ , and $\overline{\hat{C}}$  be the fundamental vector fields representing $\frac{1}{i\hbar}[\obs{A},\rho]$, $\frac{1}{i\hbar}[\obs{B},\rho]$, and $\frac{1}{i\hbar}[\obs{C},\rho]$ respectively. Then we have
\begin{eqnarray}
 \nonumber
 d\omega(\overline{\hat{A}},\overline{\hat{B}},
\overline{\hat{C}}) &=& \frac{1}{3} (\overline{\hat{A}}\cdot \omega(\overline{\hat{B}},\overline{\hat{Z}})- \overline{\hat{B}}\cdot \omega(\overline{\hat{A}},\overline{\hat{C}})+
\overline{\hat{C}}\cdot \omega(\overline{\hat{A}},\overline{\hat{B}})\\\nonumber &+&
\omega([\overline{\hat{A}},\overline{\hat{B}}],\overline{\hat{C}})+
\omega([\overline{\hat{B}},\overline{\hat{C}}],\overline{\hat{A}})+
\omega([\overline{\hat{C}},\overline{\hat{A}}],\overline{\hat{B}}))=0
\end{eqnarray}
since the last three terms vanish by the Jacobi identity and the first three terms also vanish by invariance of the symplectic form $\omega$.
\end{proof}
The importance of this form stems from the fact that
if $A$ is the expectation value function of a Hermitian operator $\obs{A}$, that is $A(\rho)=\Tr(\rho \obs{A})$, and $X_A$ is the Hamiltonian vector field associated with $A$, which is implicitly defined by the identity $dA(X)=\omega(X_A,X)$, then
\begin{equation}
X_A(\rho)=\frac{1}{i\hbar}[\hat{A},\rho].
\end{equation}
Now, $(\omega,J)$ is a K\"ahler structure, and we define $h$ to be the associated Hermitian inner product,
\begin{equation}
h(X,Y)=\omega(X,JY)+i\omega(X,Y)
\end{equation}
\begin{theorem}
Let  $\obs{A}$ and $\obs{B}$ be two observables on the Hilbert space
which are off-diagonal at $\rho$. Then we have
\begin{equation*}
h(X_A(\rho),X_B(\rho))=\frac{2}{\hbar}\sum_{i>j} (p_i-p_j)\Tr(A_{ij}^\dagger B_{ij}),
\end{equation*}
where $A_{ij}$ and $ B_{ij}$ are elements of $\obs{A}$ and $\obs{B}$  respectively.
\end{theorem}
\begin{proof}
To prove this theorem we note that $h(X_A(\rho),X_B(\rho))$ can be written as
\begin{equation*}
\begin{split}
h(X_A(\rho),X_B(\rho))
&=\frac{1}{i\hbar}\Tr\left([\obs{A},\check{B}]\rho\right)+\frac{1}{\hbar}\Tr\left([\obs{A},\obs{B}]\rho\right)\\
&=\frac{1}{\hbar}\Tr\left([\obs{A},(\obs{B}-i\check{B})]\rho\right)
\end{split}
\end{equation*}
Now, $\obs{B}-i\check{B}$ is represented by the upper diagonal matrix
\begin{equation*}
\obs{B}-i\check{B}=2\left[
  \begin{array}{ccccc}
  \0_{n_1}        & B_{12}         & B_{13}  & \ldots & B_{1k}  \\
                  & \0_{n_2}        & B_{23}  & \ldots & B_{2k}  \\
                  &                 & \0_{n_3} & \ldots & B_{3k}  \\
                  &                &           & \ddots & \vdots  \\
                  &                &                  &    & \0_{n_k}
  \end{array}
\right]
\end{equation*}
which after some straightforward calculation give the following expression for the commutation between  $\obs{A}$ and $\obs{B}-i\check{B})$
\begin{equation*}
[\obs{A},\obs{B}-i\check{B})]=2{\small\left[
  \begin{array}{ccccc}
  -\sum\limits_{j>1}B_{1j}A_{1j}^\dagger   & *         & *  & \ldots & *  \\
  *                         & A_{12}^\dagger B_{12}-\sum\limits_{j>2}B_{2j}A_{2j}^\dagger         & *  & \ldots & *  \\
*                  &         *        &  \sum\limits_{j<3}A_{j3}^\dagger B_{j3}-\sum\limits_{j>3}B_{3j}A_{3j}^\dagger    & \ldots & *  \\
  \vdots               &        \vdots   &     \vdots       & \ddots & \vdots  \\
  *                &           *     &         *         &    & \sum\limits_{j<k}A_{jk}^\dagger B_{jk}
  \end{array}
\right]}
\end{equation*}
Note that the stars represent expressions whose explicit forms need not be known. Thus we have
\begin{equation*}
h(X_A(\rho),X_B(\rho))=\frac{2}{\hbar}\sum_{i<j} (p_i-p_j)\Tr(A_{ij}^\dagger B_{ij}).
\end{equation*}
This end up the prove of our theorem.
\end{proof}
The above result is very important in proof of a geometric uncertainty relation for mixed quantum states which we will consider in the following section.

\section{Geometric uncertainty relation based on K\"{a}hler structure}\label{s2}
In this section we derive a geometric uncertainty relation for mixed quantum states based on the geometric frame work we have introduced in the pervious section.
\\
Let $\obs{A}$ be a observable on $\HH$, and  consider the  uncertainty function
\begin{equation}\label{uf}
\Delta A(\rho)=\sqrt{\Tr(\rho\obs{A}^2)-\Tr(\rho\obs{A})^2}.
\end{equation}
Now we will state the main result of this section in form of the following theorem.
\begin{theorem}
Let $\obs{A}$ and $\obs{B}$ be two obervables  on $\HH$.
Then we have
\begin{equation}\label{gur}
\Delta A\Delta B\geq \frac{\hbar}{2}|h(X_A,X_B)|.
\end{equation}
\end{theorem}
\begin{proof}
To prove the theorem we first pick a $\rho$ and fix a basis, so that $\rho=\mathrm{diag}(\rho_{1},\rho_{2},\ldots,\rho_{k})$. Then  the observable $\obs{A}$ has the following representation
\begin{equation}
\obs{A}=\left[
  \begin{array}{cccc}
    A_{11} & X_{12} & \cdots & X_{1k} \\
    X^{\dagger}_{12} & A_{22} &\cdots & X_{2k} \\
    \vdots & \vdots & \ddots & \vdots \\
    X^{\dagger}_{1k}  & X^{\dagger}_{2k}  & \cdots & A_{kk}  \\
  \end{array}
\right].
\end{equation}
Then it is not difficult to derive the following expression for our density operator and observable
\begin{eqnarray}
  \Tr(\rho\obs{A}^2)&=&\sum_{i=1}^k p_i\Tr(A_{ii}^2)+\sum_{i<j}(p_i+p_j)\Tr(X_{ij}^\dagger X_{ij}),\\
  \Tr(\rho\obs{A})&=&\sum_{i=1}^k p_i\Tr(A_{ii}).
\end{eqnarray}
Now by inserting these relations into  the equation (\ref{uf}) we get
\begin{equation}
  \begin{split}
    \Delta A(\rho)^2
    &=\sum_{i=1}^k p_i\Tr(A_{ii}^2) - \left(\sum_{i=1}^k p_i\Tr(A_{ii})\right)^2+\sum_{i<j}(p_i+p_j)\Tr(X_{ij}^\dagger X_{ij})\\
    &=(\Delta \obs{A}^\bot)^2 + \sum_{i<j}(p_i+p_j)\Tr(X_{ij}^\dagger X_{ij})\\
    &\geq \sum_{i<j}(p_i-p_j)\Tr(X_{ij}^\dagger X_{ij})\\
    &=\frac{\hbar}{2}h(X_A(\rho),X_A(\rho)),
  \end{split}
\end{equation}
where we have decomposed $\obs{A}$ as  $\obs{A}=\obs{A}^{\|}+\obs{A}^{\bot}$ and $\Delta \obs{A}^\bot=\mathrm{Tr}(\obs{A}^{\bot}\rho)$. Similarly we get $\Delta B(\rho)^2\geq\frac{\hbar}{2}h(X_B(\rho),X_B(\rho))$.
Thus,
\begin{equation}
\begin{split}
\Delta A(\rho)^2\Delta B(\rho)^2
&\geq \frac{\hbar^4}{4}h(X_A(\rho),X_A(\rho))h(X_B(\rho),X_B(\rho))\\
&\geq \frac{\hbar^4}{4}|h(X_A(\rho),X_B(\rho))|^{2},
\end{split}
\end{equation}
where in the last step we have used the Schwarz inequality. By taking the square root of both sides of this equation we get(\ref{gur}).
This end the proof of  our geometric uncertainty relation for mixed quantum states.
\end{proof}
Our geometric uncertainty relation are related to Robertson-Schr\"{o}dinger  uncertainty relation \cite{Robertson_1929}.

\section{Conclusion}
In this paper we have introduced a geometric framework for mixed quantum states based on a K\"{a}hler structure. We have explicitly defined the compatible triplet for our quantum phase space, namely a symplectic form, a Riemannian metric, and an almost complex structure. We have argued that our almost complex structure is integrable since the Nijenhuis tensor vanishes which also implies that our quantum phase space is a K\"{a}hler  manifold. Finally we have applied our geometric framework to a quantum system with two observables in order to derive a geometric uncertainty relation for quantum assembles.  Our framework can be extended to the infinite dimensional case but this issue needs further investigation. The advantages of the geometric framework is its simplicity and effectiveness. We also believe that the geometric framework can be applied and tested for different quantum systems which also could give rise to very insightful results about  quantum mechanics with many applications in the fields of quantum information, quantum computing, and quantum control.
\begin{flushleft}
\textbf{Acknowledgments:} The author acknowledges useful comments and also discussions with Ole Andresson.
The  author also acknowledges the financial support from the Swedish Research Council (VR).
\end{flushleft}

\end{document}